\documentclass[a4paper,11pt]{amsart}
\calclayout

\usepackage{bbm}
\usepackage{amssymb}
\usepackage{amsmath}
\usepackage{amsfonts}
\usepackage{mathrsfs}
\usepackage{amsthm}
\usepackage{xspace}
\usepackage{comment}
\usepackage{booktabs}
\usepackage{graphicx}
\usepackage{rotating}
\usepackage{ifpdf}
\usepackage{mdwlist}
\usepackage{verbatim}
\usepackage{ifthen}
\usepackage{hyperref}

\def\ket #1{\vert #1\rangle}
\def\bra #1{\langle #1\vert}
\def\ketbra #1 #2 {\vert #1\rangle \langle #2 \vert}
\def\braket #1 #2 {\langle #1\vert #2 \rangle}
\def\tr #1 {\trace #1}

\newcommand{\1}{\mathbbm{1}}
\DeclareMathOperator{\trace}{tr}
\newcommand{\E}{\mathbb{E}}

\newcommand{\B}{B}
\newcommand{\Hil}{H}
\newcommand{\Wg}{Wg}

\newtheorem{theorem}{Theorem}[section]
\newtheorem{lemma}[theorem]{Lemma}
\newtheorem{proposition}[theorem]{Proposition}

\newenvironment{definition}[1][Definition]{\begin{trivlist}
\item[\hskip \labelsep {\bfseries #1}]}{\end{trivlist}}
\newenvironment{remark}[1][Remark]{\begin{trivlist}
\item[\hskip \labelsep {\bfseries #1}]}{\end{trivlist}}

\begin{document}
\bibliographystyle{alpha}

{
\begin{center}
\Large\textbf{Matrix Product States, Random Matrix Theory and the Principle of Maximum Entropy}
\vspace{1em}

\large
Ben\^oit Collins$^{1,2}$, Carlos E. Gonz\'alez-Guill\'en$^{3,5}$ and David P\'erez-Garc\'ia$^{4,5}$\\
\footnotesize
\vspace{.5em}
$^1$D\'epartement de Math\'ematique et Statistique, Universit\'e d'Ottawa, 
K1N6N5 Ottawa, Canada.\\
$^2$CNRS, Institut Camille Jordan Universit«e Lyon 1, 
69622 Villeurbanne, France.\\
$^3$Departamento de Matem\'aticas, E.T.S.I. Industriales, Universidad Polit\'ecnica de Madrid, 
28006 Madrid, Spain.\\
$^4$Departamento An\'alisis Matem\'atico, Universidad Complutense de Madrid, 28040 Madrid, Spain.\\
$^5$IMI, Universidad Complutense de Madrid,  28040 Madrid, Spain.
\end{center}
\vspace{1.5em}
\begin{quote}
\small
\textbf{Abstract}
Using random matrix techniques and the theory of Matrix Product States we show that reduced density matrices of quantum spin chains have generically maximum entropy.
\vspace{1.5em}\\
\end{quote}
}

{\it Jaynes' principle of maximum entropy} \cite{Jaynes57a,Jaynes57b} gives a pretty satisfactory solution to the old problem of dealing with prior information in probability theory. Generalizing  the old {\it  principle of indifference} of Laplace, it briefly states that among all possible probability distributions compatible with our prior information, the best choice is the one which maximizes the Shannon entropy. Apart from its important applications on decision theory, since its origin it has succeeded in giving a very useful information-theoretical view of statistical mechanics, both classical \cite{Jaynes57a} and quantum \cite{Jaynes57b}--where the function to maximize is the von-Neumann entropy. As an easy illustration, given the average energy of a quantum system as prior information, the density matrix which maximizes entropy is exactly the thermal state associated to that particular energy. This spirit has been recently recovered with great success in \cite{PSW06} and further developed in a number of ways in \cite{Brandao11, Cramer11, Popescu09, Popescu10}.

To which extent the principle of maximum entropy can be extended to more and more general situations has been a very active and controversary field in the last half-century. For instance, very recently a series of theoretical and experimental works \cite{Eisert08a,Eisert08b,Bloch} seem to validate the principle in relaxation processes of quantum systems when focusing on a particular small subsystem --which, as argued in Figure \ref{fig1}, is the most relevant situation.

Another, even older, principle to assign prior probabilities in physical problems relies on the symmetries of the problem (see \cite{Jaynes67} for a discussion). For instance, if one wants to incorporate in the problem some invariance, i.e. independence of the reference frame, this already reduces the class of prior probability distributions available. Indeed, if one has enough symmetries --they form a compact group-- there is indeed a unique way of defining a prior distribution compatible with the symmetries --the Haar measure-- and the problem is solved.

But what if one wants to incorporate to the problem some less {\it standard} knowledge? For instance, that the interactions in our model are local and homogeneous and that we work at zero temperature, but not any assumption on the particular interactions in the model itself. Note the difference with Jaynes' approach where the particular Hamiltonian of the model is known. Is there any way of incorporating this information to the problem? Which is then the right prior probability? Is it related to maximizing some entropy? Since this type of assumptions are natural and widely accepted, solving these questions could be of upmost importance in quantum condensed-matter problems. In this paper we attack (and to some extent solve) them in the particular case of 1D spin systems.

\begin{figure}[t]
\includegraphics[width=15cm]{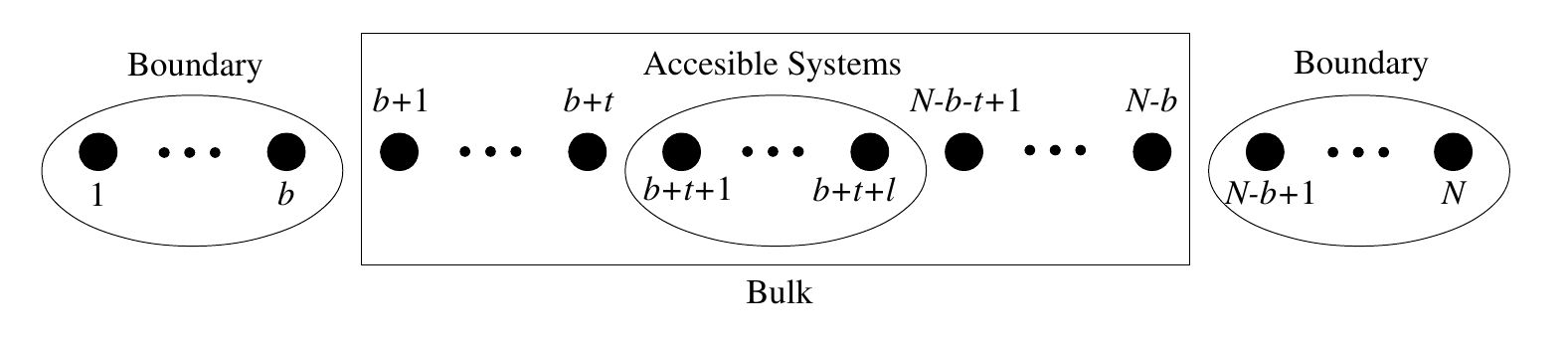}\\
\caption{In this work we consider a chain of $N$ sites, with homogeneous interactions in the bulk and boundary effects in exponentially small regions of size $b$ at the borders. We assume that the experimentally accesible region (and hence the region we are interested in) is an exponentially small region of size $l$ in the center of the chain.\label{fig1}}
\end{figure}

To do that we will take advantage of the recent developments in the understanding of quantum spin chains, where it is nowadays widely well justified, both numerically \cite{White92} and analytically \cite{Hastings07}, that their ground states are exactly represented by the set of Matrix Product States (MPS) with {\it polynomial} bond dimension. We will concentrate in the situation of a chain with boundary effects in exponentially small regions of size $b$ at both ends, homogeneity in the bulk and experimental access to an exponentially small central region of size $l$ (see Figure \ref{fig1}). Tracing out the boundary terms leads to a bulk state given by \begin{equation}\label{eq:inicio}
\rho=\sum_{i_{b+1},...i_{N-b},j_{b+1},...j_{N-b}=1}^d \tr (L A_{i_{b+1}} \cdots A_{i_{N-b}} R  A^\dagger_{j_{N-b}} \cdots A^\dagger_{j_{b+1}} )\ket {i_{b+1}...i_{N-b}} \bra {j_{b+1}...j_{N-b}},
\end{equation}
where all $A_i$, $L\ge 0$ and $R\ge 0$ are $D\times D$ matrices with $D={\rm poly}(N)$. This will be our starting point, that is, the prior information can be
understood as restricting the bulk-states of our system as having the form (\ref{eq:inicio}).

Now, it is also known from the general theory of MPS \cite{PVWC06} that this set has a natural (over)para\-metrization by the group $U(dD)$, via the map $U\mapsto A_i=\bra 0 U \ket i$. Being $U(dD)$ a unitary group, one can use the symmetry-based assignment of prior distributions to sample from the Haar measure. Similarly, the fact that the map $X\mapsto \sum_i A_iXA_i^\dagger$ is trace-preserving leads to consider $\tr(R)=1$, $\|L\|_\infty\le 1$, giving us natural ways of sampling also the boundary conditions (see below).
One can therefore ask about which is then the generic reduced density matrix $\rho_l$ of $l\ll N$ sites. Note that, by the above comments, this is nothing but asking about generic observations of 1D quantum systems. This idea has been already exploited for the non-translational invariant case in \cite{GOZ10}. The aim of the present work is to show that $\rho_l$ has generically maximum entropy:

\begin{theorem} \label{thetheorem}
Let $\rho_l$ be taken at random from the ensemble introduced with $D\geq N^{5}$. Then
$\| \rho_l/\trace \rho_l-\1/d^l\|_{\infty} \leq (d^l-1)\sqrt{d^l}O(D^{-1/10})$  except with probability exponentially small in D.
\end{theorem}

Note that, since the accessible region $l$ is exponentially smaller than the system size, the bound can be made arbitrary small while keeping the size of the matrices $D$ polynomial in the system size.

To prove the theorem, we will rely on recent developments of random matrix theory, in particular on the graphical Weingarten calculus provided in \cite{CN10}, and on a novel estimate of the Weingarten function.

The paper is organized as follows. First we introduce the Matrix Product State formalism, then we introduce Weingarten function and calculus, then we introduce basic results of the concentration of measure phenomenon. Finally, in section \ref{sec4}, we prove theorem \ref{thetheorem} using the tools already introduced together with a novel asymptotic bound of the Weingarten function.

\begin{section}{Random Matrix Product States}\label{sec:MPS}
In this section we just fix the notation, for a detailed exposition see \cite{PVWC06}. Let $\dim (H_A) =D$ and $\dim (H_B)=d$, our initial state $\rho$ given by Equation (\ref{eq:inicio}) can be expressed by means of the map $\mathcal E(X): \B (\Hil_A) \rightarrow \B (\Hil_A\otimes \Hil_B)$ given by $\mathcal E(X)=\sum \limits_{i,j} A_i X A_j^\dagger \otimes \ket i\bra j$, simply as
$$ \rho = tr_A[L \mathcal E^{n}(R)], $$\\
where one should understand the map acting only in A and {\it creating} the systems B in order from 1 to n.

In the rest of this work we will be interested in the reduced state of the $l$ consecutive central sites of the chain, where $l<<n=2t+l$, that, up to normalization, will be described as
$$\rho_l = tr_{A,B_1...B_{t}, B_{t+l+1}...B_n} [L \mathcal E^{n}(R)]. $$
The general boundary conditions $L$ and $R$ come from tracing out the boundary sites as described in figure $\ref{fig1}$. MPS theory leads to consider them belonging respectively to the sets $\mathcal L=\{L\geq 0: \|L\|_\infty \leq 1, L \in M_{D}\}$ and $\mathcal R=\{R\geq 0: \tr(R)= 1, R \in M_{D}\}$, where $\|\cdot\|_\infty$ means the usual operator norm. Diagonalizing $L=V \Lambda V^\dagger$ and $ R=W \Omega W^\dagger $ we can parametrize $\mathcal{L}$ by $[0,1]^D \times U(D)$ and $\mathcal{R}$ by $\mathbbm S_1([0,1]^D)\times U(D)$, being $\mathbb S_1([0,1]^D)$ the set of $D$-event probability distributions. Again by symmetry considerations, this leads to sample $L$ using the Lebesgue measure on $[0,1]^D$ and the Haar measure on $U(D)$ and to sample $R$ using any permutational invariant measure on  $\mathbb S_1([0,1]^D)$ and the Haar measure on  $U(D)$. We finally recall from the introduction that the matrices $A_i$ in Equation (\ref{eq:inicio}) will be sampled from the Haar measure on $U(Dd)$ via the parametrization $U\mapsto A_i=\bra 0 U \ket i$.

Summarizing, we take the ensemble of MPS defined by the tuple $(U,L,R)=(U,V,W,\Lambda,\Omega)$ where $U$, $V$ and $W$ are distributed with respect to the Haar measure in the respective unitary group, $\Lambda$ is distributed according to the Lebesgue measure in $[0,1]^D$ and $\Omega$ according to any permutation invariant probability measure in $\mathbbm S_1 ( [0,1]^D)$.

Now our problem can be rephrased as:

\

{\it Given  $(U,L,R)$ randomly chosen find the behavior of the normalized state corresponding to
$$\rho_l (U,L,R)  = \trace_{A,B_1...B_t, B_{t+l+1}...B_n} (L U_{A,B_1}\cdot\cdot\cdot U_{A,B_n} (R \otimes (\ket 0 \bra 0)^{\otimes n}) U_{A,B_n}^{ \dagger}\cdot\cdot\cdot U_{A,B_1}^{\dagger}).$$
where the systems of the sites $B_i$ are labeled from left to right and the unitary matrices are acting in the site indicated and the ancillary system A from right to left living the other sites invariant.}

\

This state is represented in the graphical level in figure \ref{figint}.

\begin{figure}[t]
\includegraphics[width=15.5cm]{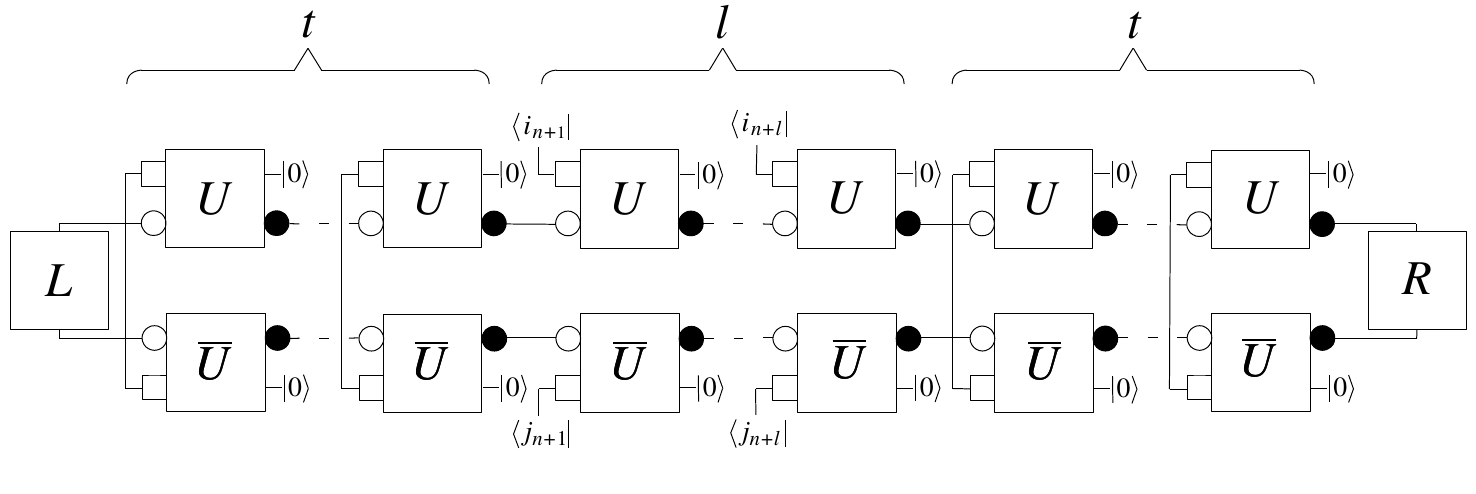}\\
\caption{Graphical representation of $\rho_l$, where big squares represent matrices and the small objects attached to them represent the tensors that form the matrices. Dark objects correspond to "ket" tensors, white objects correspond to "bra" tensors, squares are used for dimension d and circles for dimension D. Wires represent contraction rules between tensors.}\label{figint}
\end{figure}

\end{section}

\begin{section}{Weingarten Function and Calculus} \label{Weing}
The Weingarten function was first introduce in \cite{Weingarten78}, for a complete description of this function we refer to \cite{Collins03}. Here we just describe its main ingredients to focus on the graphical calculus introduced in \cite{CN10}. We will follow the standard notation of representation theory of symmetric groups.  We denote by ${\lambda \vdash p}$ that $\lambda$ is a partition of $p$, $\chi^\lambda$ is the corresponding character of $S_p$ and $s_{\lambda,d}(x)=s_{\lambda,d}(x,...,x)$ the corresponding Schur function, see \cite{Fulton97}. If $\sigma\in S_p$ we denote by $|\sigma|$ the minimum number k such that $\sigma$ can be written as a product of $k$ transpositions, $\#\sigma$ is the number of cycles in $\sigma$ and both quantities are related by the formula $|\sigma|=p-\#\sigma$.
\begin{definition}\label{defWein}
The Weingarten function $\Wg(n,\sigma)$ takes as inputs a dimension parameter $n$ and a permutation $\sigma$ in the symmetric group $S_p$ and is given by
$$\Wg(n,\sigma)=\frac {1} {p!^2} \sum_{\lambda \vdash p} \frac  {(\chi^\lambda (1))^2 \chi^\lambda (\sigma)} {s_{\lambda,n}(1)}.$$
\end{definition}

Its importance relys on the following theorem from \cite{Collins03}, which tells us that the average of a monomial over the unitary group can be computed in terms of sums of Weingarten functions.

\begin{theorem}
Let n be a positive integer and $i=(i_1,...,i_p)$, $i'=(i'_1,...,i'_p)$, $j=(j_1,...,j_p)$ and $j'=(j'_1,...,j'_p)$ be p-tuples of positive integers from ${1,2,...,n}$. Then

$$\int_{\mathcal U(n)} U_{i_1j_1}... U_{i_pj_p} \overline {U_{i'_1j'_1}}...\overline{U_{i'_pj'_p}} dU=\sum_{\sigma,\tau \in S_p} \delta_{i_1i'_{\sigma(1)}}... \delta_{i_pi'_{\sigma(p)}}\delta_{j_1j'_{\tau(1)}} ... \delta_{j_pj'_{\tau(p)}} \Wg(n,\tau\sigma^{-1}).$$

\end{theorem}

In \cite{CN10} the authors introduce a graphical paradigm in order to simplify the computation of the average of polynomials over the unitary group. Consider a polynomial $P(U)$ of degree $p$ in $U$ and $\overline{ U}$ then 
\begin{equation}\label{Graph}
\mathbbm E_U(P(U))=\sum_{\sigma,\tau \in S_p } C_{(\sigma,\tau)} \Wg(n,\tau\sigma^{-1}),
\end{equation}
where the coefficients $C_{(\sigma, \tau)}$ can be computed by the following procedure in figures \ref{figint}, \ref{figWeing1} and \ref{figWein2}. One has to enumerate the matrices $U$ and $\overline U$ respectively from $1$ to $p$ and for any two permutations $\sigma, \tau \in S_p$ we delete the $U$ and $\overline U$ boxes and we connect the white square and circle in $U_i$ with the white square and circle respectively in $\overline {U_{\sigma(i)}}$, and analogously with the black objects and the $\tau$ permutation. Now, loops represent traces over the matrices involved in it. If there is no matrix involved in a loop then it represents the trace of the identity of the system. Finally, if there are paths that are not loops they translate into the contraction with the boundary conditions that appear in it.
The number $C_{(\sigma, \tau)}$ is just the product of the values of all the contractions. Note that, as drawn in figure \ref{figint}, a monomial in $U^\dagger$ can be substituted by a monomial in $\overline U$.

\end{section}

\begin{section}{Measure of Concentration Phenomenon}

In this section we introduce the basic results of the measure concentration phenomenon that we are going to use; for a detailed exposition see for example \cite{MS86, Ledoux01}.

\begin{definition}
Let $(X,d)$ be a metric space with probability measure $\mu$, its concentration function  $\alpha_{(X,d,\mu)}$ is defined as
$$\alpha_{(X,d,\mu)}(r)=\sup \{1-\mu(A_r);A\subset X,\mu(A)\geq \frac 1 2\} \text{, $r>0$}$$ where $A_r=\{x\in X;d(x,A)<r\}$ is the open r-neighbourhood of A (with respect to d).
\end{definition}

This definition allows to prove directly that almost all the images of a Lipschitz function concentrate around the median, where the concentration factor is given by the concentration function. Nevertheless, we are interested in the concentration around the mean which is a consequence of the other, as one can bound the distance between the median and the mean depending on the concentration function.

\begin{theorem}[Measure concentration phenomenon] \label{phenomenon}
Let $F$ be a Lipschitz function on $(X,d)$, and $\mu$ a probability measure on $(X,d)$, then
$$\mu(\{F>\mathbbm E_\mu(F)+r\})\leq 2\alpha_\mu(r/\|F\|_{Lip}),$$
$$\mu(\{F<\mathbbm E_\mu(F)-r\})\leq 2\alpha_\mu(r/\|F\|_{Lip}),$$
where $\mathbbm E_\mu(F)$ is the mean of $F$ with respect to $\mu$ and $\|F\|_{Lip}$ is the Lipschitz constant of $F$.
\end{theorem}

When one is interested in the concentration properties of a family of spaces, what matters is the scaling of the concentration function depending on the parameter defining the family of spaces. Thus looking at the definition of concentration function one can prove that the concentration properties of two spaces behave at least as well as the worst one of the two.
\begin{proposition}
Let  $\mu$, $\nu$ two probability measures on metric spaces $(X,d)$ and $(Y,\delta)$ respectively. Then, if $\mu \times \nu$ is the product measure in $X \times Y$ equipped with the $l^1$-metric, $\alpha_{(X\times Y,d+\delta,\mu\times \nu)}\leq \alpha_{(X,d,\mu)}+ \alpha_{(Y,\delta,\nu)}$
\end{proposition}

If we apply this proposition to the spaces we are interested in we have the following lemma.

\begin{lemma}
Let $\mu$ be the Haar measure in $(U(D),d_2$), the unitary group with the Hilbert-Schmidt distance. Let $\nu$ be the Lebesgue measure in $([0,1]^D,d_\infty)$, the hypercube with the maximum distance.
Then, for any $ k \in \mathbbm N$, the product space $(X,d_1)$, where $X=U(k D)\times U(D) \times U(D) \times [0,1]^D$ and $d_1$ is the $l_1$ distance of the product space, with the product probability measure $\eta=\mu \times \mu \times \mu \times \nu$ has concentration function
$$\alpha_{(X,\delta,\eta)}(r)\leq c e^{-C n r^2 },$$
where c and C are universal constants.
\end{lemma}

\begin{remark}
Note that in this lemma, that will be used to prove the concentration in theorem \ref{main}, we do not consider the space $\mathbbm S_1([0,1]^D)$, as in all the theorems below the result holds independently of $\Omega \in \mathbbm S_1([0,1]^D)$. \end{remark}

\end{section}

\begin{section}{Proof of the main Theorem}\label{sec4}

Recall that we are considering the ensemble of MPS defined by the tuple $(U,L,R)=(U,V,W,\Lambda,\Omega)$ where $U$, $V$ and $W$ are distributed with respect to the Haar measure in the respective unitary group, $\Lambda$ is distributed according to the Lebesgue measure in $[0,1]^D$ and $\Omega$ according to any permutation invariant probability measure in $\mathbbm S_1 ( [0,1]^D)$. To prove Theorem \ref{thetheorem} we need to compute the mean and the Lipschitz constant of the trace normalized version of $f(\rho)=\tr (\rho^2_l (U,L,R))$ over the introduced ensemble. The difficulty of this calculus comes from computing the mean of the function. To simplify our computations we will first compute the mean and the Lipschitz constant for both function $f(\rho)$ and its normalization function $g(\rho)=(\tr \rho_l (U,L,R))^2$ and then argue about the concentration of $\tr(\rho_{Norm}^2)=f(\rho)/g(\rho)$.  To compute the mean of $f(\rho)$ we first need to give a novel asymptotic bound of the Weingarten function.

\begin{theorem}\label{bound}
Let $p$, $n$  and $k$ be nonnegative integers such that $p^k\leq n$.
Then there exists a constant K depending only on k such that for any $\sigma\in S_p$,
$$\Wg(n,\sigma) \leq K n^{-p-|\sigma|(1-2/k)}.$$
\end{theorem}

\begin{proof}
We recall, see \cite{Fulton97}, that for any partition $\lambda \vdash p$ of the integer $p$,
$$s_{\lambda,n}(1)=\frac {\chi^\lambda (1)}{p ! }\prod_{i=1}^p (n-\lambda_i)$$
where $\lambda_i$ is an integer in $\{0,...,p-1\}$. Equivalently, the Weingarten function becomes
\begin{equation}\label{weingeq}
\Wg(n,\sigma)=\frac {1} {n^p p!} \sum_{\lambda \vdash p} \frac {\chi^\lambda (1) \chi^\lambda (\sigma)} {\prod_{i=1}^{p} (1-\lambda_i/n)}.
\end{equation}
Consider the function
$$f_\lambda: z \rightarrow (\prod_{i=1}^p (1-z\lambda_i))^{-1}$$
This function is holomorphic in a neighborhood of zero. Moreover $2\leq k$ , since $p^2\leq n$, we have, for any $|z|\leq p^{-2}$,
$$|f_\lambda(z) | \leq e.$$

As a consequence, writing $f_\lambda(z)=\sum_{i\geq 0} a_{i,\lambda} z^i$, we obtain the Cauchy estimate
$$a_{i,\lambda}\leq e p^{2i}.$$
But equation \ref{weingeq} implies that
$$\Wg(n,p,\sigma)=\frac {1} {n^p p!} \sum_{\lambda \vdash p} \chi^\lambda (1) \chi^\lambda (\sigma) (1+\sum_{i\geq 1} a_{i,\lambda} n^{-i}).$$
Therefore the coefficient in $n^{-p-i}$ has norm smaller than $e p^{2i}$. But all coefficients are zero until $i=|\sigma|$ \cite {Collins03}, therefore
$$\Wg(n,\sigma)\leq n^{-p-|\sigma|} e(1+\frac {p^2}{n}+(\frac {p^2}{n})^2+...)p^{2|\sigma|}.$$
For $n\geq 2$,  and since $p^k\leq n$, this implies $p^{2|\sigma|}\leq n^{2|\sigma|/k}$. Furthermore $(1+\frac {p^2}{n}+(\frac {p^2}{n})^2+...)$ can be bounded by a universal constant (5, for example). The result follows.

\end{proof}

To organize the computations and the reasoning used in the bound of the mean of $f(\rho)$ we prove the following two lemmas.

\begin{lemma}\label{gamma}
Let $\alpha, \beta, \gamma \in S_{p}$,

a) the quantity $|\gamma ^{-1} \alpha \gamma \alpha^{-1} \beta|+|\beta| $ is an even number.

b) If $\gamma ^{-1} \alpha \gamma \alpha^{-1} \beta=c$ and $\gamma ^{-1} \alpha' \gamma \alpha'^{-1} \beta=c$, then $\alpha'^{-1}\alpha$ commutes with $\gamma$.

c) If $p=2n+4$ and $\gamma=(2n+1,1, 2,..., n,2n+3)(2n+2,n+1,n+2,....,2n,2n+4)$ and $\alpha(2n+i)=2n+i$ for $i=1,...,4$. Then the function that takes $(\alpha ,\beta)$ to $(g,h)$ with $g=\gamma^{-1} \alpha \gamma \alpha^{-1}$ and $h =\beta \alpha^{-1}$ is one to one.

\end{lemma}

\begin{proof}
a) The result follows from the fact that the parity of $|\alpha \beta|$ is the same as the parity of $| \alpha|+|\beta|$.

b)  If $\gamma ^{-1} \alpha \gamma \alpha^{-1} \beta=c$ and $\gamma ^{-1} \alpha' \gamma \alpha'^{-1} \beta=c$, then $\alpha \gamma \alpha^{-1}=\alpha' \gamma \alpha'^{-1}$. Multiplying by the inverse of the right hand side we have that $\gamma^{-1}\alpha'^{-1}\alpha \gamma \alpha^{-1}\alpha'=\1$ which happens if and only if $\alpha'^{-1}\alpha$ commutes with $\gamma$.

c) If $\alpha$ is fixed, the change $h=\beta \alpha^{-1}$ is clearly one to one. Now, by b), the change $g=\gamma^{-1} \alpha \gamma \alpha^{-1}$ is one to one if and only if $\gamma^{-1} \alpha \gamma \alpha^{-1}=\1$ has only the trvial solution $\alpha=\1$, which can be easily check to be the case by the definition of $\gamma$ and the constraints in $\alpha$.
\end{proof}

\begin{lemma}\label{averages}
Let $L \in \mathcal L$ and $R\in \mathcal R$ be given at random with respect to the measures introduced in Section \ref{sec:MPS}. For any $\Omega \in \mathbbm S_1([0,1]^D)$, we have that $\mathbbm E[ \tr (L) ]=D/2$, $\mathbbm E[ \tr (L^2) ]=D/4$, $\mathbbm E[ \tr (R) ]=1$, $\mathbbm E[ \tr (R^2) ]\leq1$, $\mathbbm E[ \tr (L R) ]=1/2$, $\mathbbm E[ \tr (L R R) ] \leq 1/2$, $\mathbbm E[ \tr (L L R) ] = 1/4$, $\mathbbm E[ \tr (L L R R)] \leq 1/ 4$ and $\mathbbm E[ \tr (L R L R) ]\leq 1/4+ 1/4D$. 
\end{lemma}

\begin{proof}
These averages are not difficult to compute directly, but they can also be computed using the graphical calculus described in section \ref{Weing} as a warming up for the forthcoming computations. We compute the last one as example.
$$\mathbbm E[ \tr (L R L R) ]=\mathbbm E[ \tr ( V \Lambda V^\dagger W \Omega W^\dagger V \Lambda V^\dagger W \Omega W^\dagger) ]=\mathbbm E[ \tr (U \Lambda U^\dagger \Omega U \Lambda U^\dagger \Omega ) ],$$
where $\Lambda \in [0,1]^D$, $\Omega \in \mathbbm S_1([0,1]^D)$ and $U,V,W \in U(D)$ and the second equation follows by the invariance of the Haar measure. Using that $\mathbbm E_{\Lambda} [\tr \Lambda]=D/2$,\hspace{0,1cm} $\mathbbm E_{\Lambda}[\tr (\Lambda^2)]=D/4$,\hspace{0,1cm} $\tr( \Omega^2)\leq 1$ and $ (\tr \Omega)^2=1$\hspace{0,1cm} together with the graphical calculus in figure \ref{figWeing1} we get
$$\mathbbm E_{L,R}[ \tr (L R L R) ]=\mathbbm E_{\Lambda,\Omega}[\sum_{\alpha, \beta \in S_2} C_{(\alpha, \beta)}\Wg(D,\alpha \beta^{-1})]=$$
$$=\mathbbm E_{\Lambda,\Omega} [((\tr \Lambda)^2 \tr( \Omega^2)+\tr (\Lambda)^2 (\tr \Omega)^2)\Wg(D,(1)(2))- ((\tr \Lambda)^2 (\tr \Omega)^2+(\tr \Lambda^2) (\tr \Omega^2))\Wg(D,(12))]$$
$$\leq 1/4+1/4D.$$
\end{proof}

\begin{figure}[t]
\includegraphics[width=15cm]{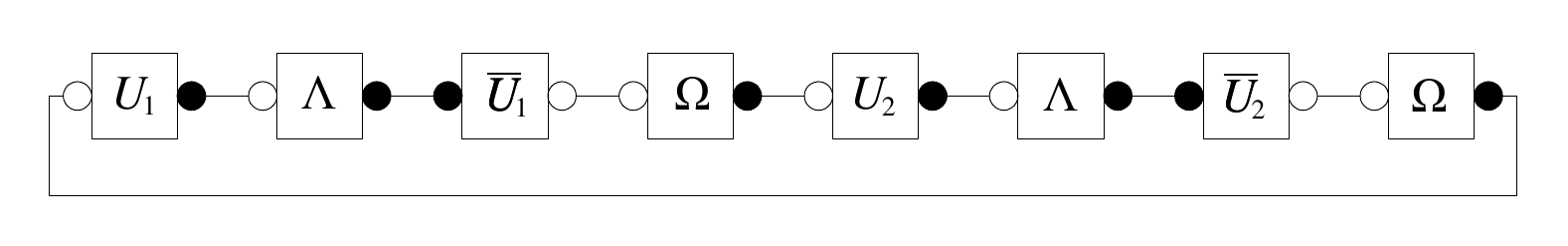}\\
\caption{Computation of $\mathbbm E_{L,R}[ \tr (L R L R) ]$ using Weingarten graphical calculus: for any two permutations $\alpha, \beta \in S_2$ delete the $U$ and $\overline U$ matrices, join the white circle of $U_i$ with the white circle of $U_{\alpha(i)}$, join the black circle of $U_i$ with the black circle of $U_{\beta(i)}$. The number $C_{(\alpha, \beta)}$ is the product of the traces involved in the new picture.}\label{figWeing1}
\end{figure}

Now, we have all the ingredients in order to compute the averages of $f(\rho)$ and $g(\rho)$.

\begin{theorem}\label{rho}
Let $\rho_l$ be taken at random from the ensemble introduced. Then, for any $\Omega \in \mathbbm S_1([0,1]^D)$ and $U\in U(dD)$, we have that
$$\mathbbm E_{\Lambda,V,W} (\tr \rho_l(\Lambda,\Omega,U,V,W))  =1/2$$
\end{theorem}

\begin{proof}
$$\mathbbm E_{\Lambda,V,W} (\tr \rho_l(\Lambda,\Omega,U,V,W)) =\mathbbm E_{\Lambda}(\tr \mathbbm E_{V,W}{[V \Lambda V^\dagger \mathcal E^n(W \Omega W^\dagger)]})=$$
$$=\mathbbm E_{\Lambda}(\tr_A \tr_{B_1,...,B_n}{[\tr(\Lambda) \frac {\1_A} D \mathcal E^n(\tr (\Omega)\frac {\1_A} D)]})=\mathbbm E_{\Lambda}(\tr_A {\tr(\Lambda) \frac {\1_A} {D^2}}) =\mathbbm E_{\Lambda} (\tr(\Lambda)/D)=1/2.$$
The first equality follows by linearity of the trace, the third because $\1_A$ is the fixed point of $tr_B(\mathcal E)$, the second and fourth just by computing the averages  themselves.
\end{proof}

Note that in this proof we are averaging over V and W, and our expression is a polynomial of degree one with respect to them. This, together with the fact that the average is independent of U make it easy to compute the average. The bound of the other function is much more involved and makes use of the asymptotic bound of the Weingarten function and Lemmas \ref{gamma} and \ref{averages}. We state it. The proof will be given in Appendix \ref{appA}.

\begin{theorem}\label{rhosquare}
Let $\rho_l$ be taken at random from the ensemble introduced and $D\geq n^5$. For any $\Omega \in \mathbbm S_1([0,1]^D)$, then 
$$\mathbbm E(\tr \rho^2_l) \leq  \frac 1 {4 d^{l}} + O(D^{-1/5}).$$
\end{theorem}

In order to apply the measure concentration phenomenon we only need to compute the Lipschitz constant of the functions we are interested in. The proof of the following theorem will be given in Appendix \ref{appB}.

\begin{theorem}[Lipschitz constants]\label{Lipschitz}
For any $\Omega \in \mathbbm S_1([0,1]^D)$, let $f(U,V,W,\Lambda)= (\trace {\rho_l(U,V,W,\Lambda, \Omega)})^2$ and $g(U,V,W,\Lambda)=\trace \rho^2_l(U,V,W,\Lambda, \Omega)$ where $$\rho_l(U,V,W,\Lambda, \Omega)  = \trace_{A,B_1...B_t, B_{t+l+1}...B_n} (V\Lambda V^\dagger U_{A,B_1}\cdot\cdot\cdot U_{A,B_n} (W \Omega W^\dagger \otimes (\ket 0 \bra 0)^{\otimes n}) U_{A,B_n}^{ \dagger}\cdot\cdot\cdot U_{A,B_1}^{\dagger}).$$
Then the Lipschitz constants of both functions are upper bounded by  $4 n+10$.

\end{theorem}

Now we can show which is the behavior of the $2$-Renyi entropy, or equivalently the purity of the normalized state $\rho_{Norm}= \rho_l/\tr(\rho_l)$. 
\begin{theorem} \label{main}
Let $\rho_l$ be taken at random from the ensemble introduced with $D\geq n^5$. Then
$\tr(\rho_{Norm}^2)=\tr \rho^2_l /(\tr \rho_l)^2=1/d^l+O(D^{-1/5})$ except\hspace{1mm} with probability exponentially small in D.
\end{theorem}

\begin{proof}

Putting together the measure concentration phenomenon \ref{phenomenon}, the bounds on the Lipschitz constant \ref{Lipschitz} and the union bound, we have for all $\Omega \in \mathbbm S_1([0,1]^D)$ that except with probability $c_1 e^{-c_2 \epsilon^2 D/n^2}$
$$\trace (\rho^2) \leq \E(\trace(\rho^2))+\epsilon \text{\hspace{1cm} and \hspace{1cm}} (\trace \rho)^2 \geq \E((\trace \rho)^2)-\epsilon$$
both at the same time where $c_1$ and $c_2$ are universal constants. Thus, we can bound
$$\frac {\trace(\rho^2)} {(\trace\rho)^2} \leq \frac {\E(\trace(\rho^2))+\epsilon} {\E((\trace\rho)^2)-\epsilon}\leq \frac {\E(\trace(\rho^2))+\epsilon} {(\E(\trace\rho))^2-\epsilon} \leq \frac{ \frac 1 {4 d^{l}} + O(D^{-1/5})+\epsilon}{1/4 -\epsilon } =\frac 1 {d^l}+ O(D^{-1/5}),$$
where the second inequality follows from Jensen's inequality, the third inequality follows from Theorem  \ref{rho} and \ref{rhosquare} and in the last equality we have used that we can take $\epsilon=O(D^{-1/5})$. The result follows.
\end{proof}

Finally we can easily prove our main theorem, which bounds the distance between the reduced density matrix of a generic random MPS and the completely mixed state.

\begin{proof} [Proof of theorem \ref{thetheorem}]\

$\rho_{Norm}=\rho_l/\trace \rho_l$ is trace normalized, that is, its eigenvalues sum up to one. Thus, in order to have an eigenvalue of $\rho$ as far as possible from $1/d^l$, the distribution of eigenvalues optimizing this problem is the one that has one eigenvalue as small or big as possible and the rest all equal. In both cases the distance between this eigenvalue and $1/d^l$ is $(d^l-1)\sqrt{d^l}O(D^{-1/10})$.
\end{proof}
\end{section}

\begin{section}{Conclusions}
In this work we have shown how reduced density matrices of  small subsystems of translational invariant random MPS have generically maximum entropy. 
This can be read as recovering Jayne's principle of maximum entropy in the situation where the prior information to incorporate in the sampling procedure is the locality and homogeneity of the interactions. For that we have relied on the (well justified) fact that MPS are the right representation for ground states of one dimensional local Hamiltonians and in the natural way of sampling MPS based on the symmetry principle.

We acknowledge Ion Nechita for very insightful discussions during the preparation of this manuscript and Mittag-Leffler institute for the organization of the 2010 fall in Quantum Information Theory where this work was initiated. C.G.G. and D.P.G.'s research was supported by EU grant QUEVADIS and Spanish projects QUITEMAD and MTM2011-26912. B.C.'s research was supported by ANR GranMa, ANR Galoisint and NSERC grant RGPIN/341303-2007.

\end{section}

\bibliography{bibRMPS}

\appendix

\begin{section}{Proof of Theorem \ref{rhosquare}}\label{appA}

\begin{figure}[t]
\includegraphics[width=15cm]{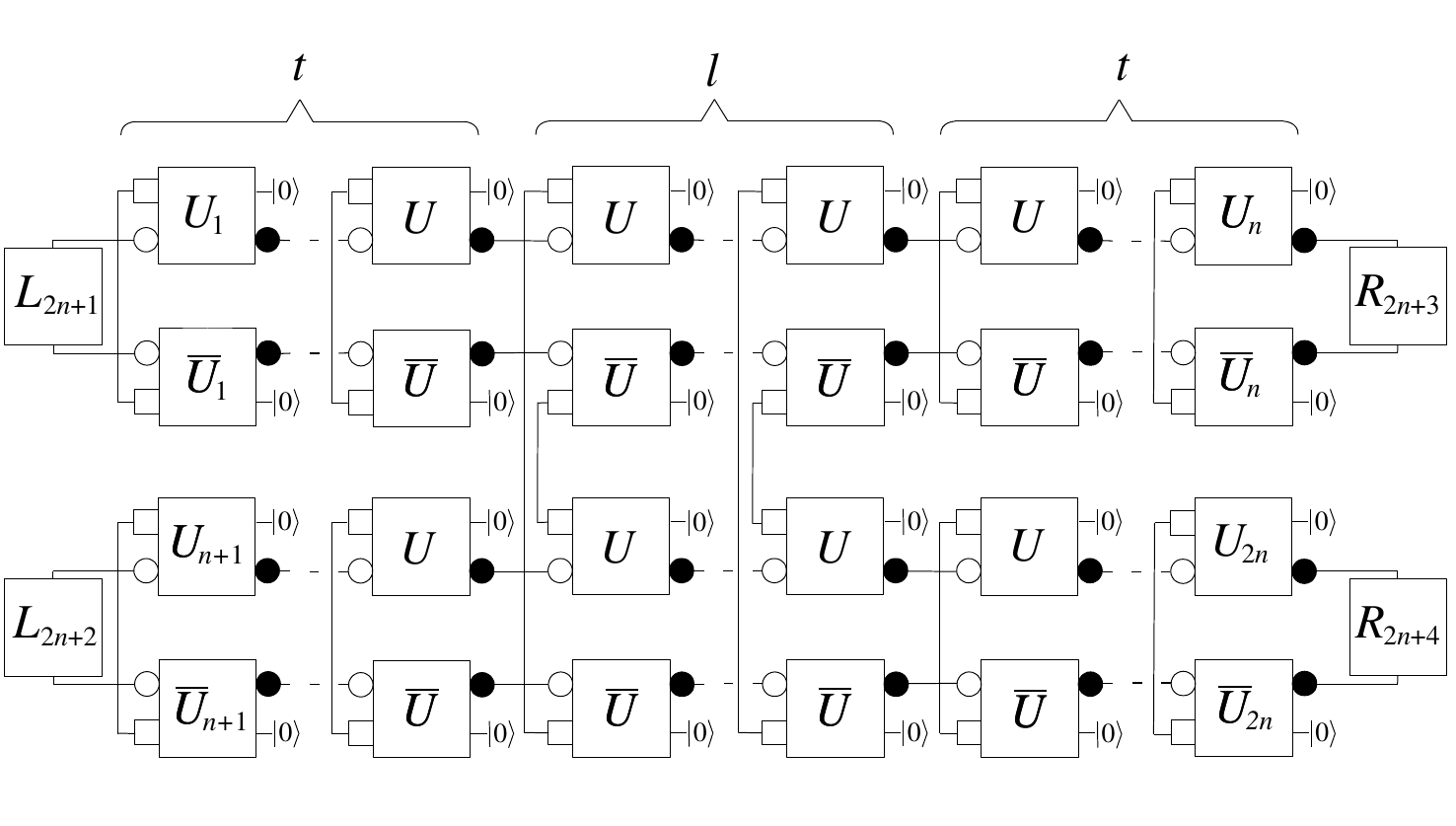}\\
\caption{Computation of $\mathbbm E_U (\tr \rho^2_l(L,R,U))$\hspace{0,1cm} using Weingarten graphical calculus: for any two permutations $\alpha, \beta \in S_{2n}$ delete the $U$ and $\overline U$ matrices, join the white circle of $U_i$ with the white circle of $U_{\alpha(i)}$, join the black circle of $U_i$ with the black circle of $U_{\beta(i)}$. The number $C_{(\alpha, \beta)}$ is the product of the traces involved in the new picture.}\label{figWein2}
\end{figure}

We can take first the average in U and using equation \ref{Graph} together with the bound of theorem \ref{bound} we have

\begin{equation*}
\begin{aligned}
\mathbbm E_{L,R,U} [\tr \rho^2_l(L,R,U)] &= \sum_{\alpha, \beta \in S_l} \mathbbm E_{L,R} [C_{(\alpha, \beta)}]\Wg(Dd,\beta\alpha ^{-1})=\\
&= \mathbbm E_{L,R} [C_{(\1, \1)}]\Wg(Dd,\1)+ \sum_{(\alpha, \beta)\neq(\1,\1) \in S_l } \mathbbm E_{L,R} [C_{(\alpha, \beta)}]\Wg(Dd,\beta\alpha ^{-1})\leq\\
&\leq \mathbbm (Dd)^{-2n}\mathbbm E_{L,R} [C_{(\1, \1)}]+K\sum_{(\alpha, \beta)\neq(\1,\1) \in S_l}  \mathbbm E_{L,R} [C_{(\alpha, \beta)}](Dd)^{-2n-3/5|\beta\alpha ^{-1}|},\\
\end{aligned}
\end{equation*}
where we are separating the case where $\alpha=\beta=\1$ and using the known value of the Weingarten function $\Wg(Dd,\1)=\mathbbm (Dd)^{-2n}$. The reason to do so is that this term is the largest one in the sum (as it will become clear through the proof).

In order to compute the coefficients $C_{(\alpha, \beta)}$ we apply the graphical Weingarten Calculus to figure \ref{figWein2}. That is, given a permutation $\alpha$ that links the white squares and circles of the $U's$ with the white squares and circles in $\overline U's$ and a permutation $\beta$ that links the black circles of the $U's$ with the black circles in $\overline U's$. We numerate the $U$ matrices from left to right and top to bottom and the same for the $\overline U$ matrices. Moreover, we enumerate the matrices $L$ as $2n+1$ and $2n+2$ and the matrices $R$ as $2n+3$ and $2n+4$.

Now the links (from left to right) between the circles of the matrices $U's$, $L$ and $R$ are given by the function $\gamma=\left(\begin{array}{ccccc}2n+1 & 1 & 2 &... & n \\1 & 2 & 3 & ... & 2n+3\end{array}\right)\left(\begin{array}{ccccc}2n+2 & n+1 & n+2 & ... & 2n \\ n+1& n+2 & n+3 & ... & 2n+4 \end{array}\right)$. One can add two extra (non-existing) links $\gamma(2n+3)=2n+1$ and $\gamma(2n+4)=2n+2$; that way $\gamma$ is a permutation. Analogously, we have the same permutation $\gamma$ for the $\overline U$ matrices. The permutation relating the links of the squares of $U$ and $\overline U$ is $\tau=(t+1,n+t+1)(t+2,n+t+1)...(t+l,n+t+l)$.
Besides, define $\alpha'= \alpha(2n+1)(2n+2)(2n+3)(2n+4)$ and $\beta'= \beta (2n+1) (2n+2)(2n+3)(2n+4)$ as the permutation $\alpha$, $\beta$ but considering it as an element of $S_{2n+4}$.

The number of loops relating the circles is  $\#  \gamma^{-1} \alpha' \gamma \beta'^{-1}-2=2n+2-|\gamma^{-1} \alpha' \gamma \beta'^{-1}|$, taking into account those where $L$ or $R$ appears. Note that we are subtracting $2$ loops $(2n+1)(2n+2)$ that we have added when including the (non-existing) links of the permutation $\gamma$. The number of loops relating the squares is $\# \tau \alpha=n-|\tau\alpha|$.
All the loops are trivial and thus they correspond to the dimension of the system, except those where L or R appears, in which we will take averages. For $\alpha=\beta=\1$, we have that $$\mathbbm E_{L,R} [C_{(\1, \1)}]=\mathbbm E_{L,R} [\tr(L)^2 \tr(R)^2 D^{2n-2} d^{2n-l}]=1/4D^{2n}d^{2n-l}.$$ 

Taking averages in $L$ and $R$ and using the bounds from lemma \ref{averages}, it can be shown, by inspection on all possible combinations of $L$ and $R$ in different loops, that it is enough to distinguish the following two cases:

\begin{enumerate}
\item $\alpha,\beta\in \mathcal{A}_{2n} $, where $\mathcal{A}_{2n}$ is the set of tuples where $2n+3$ and $2n+4$ are not in the same cycle of the permutation $ \gamma^{-1} \alpha' \gamma \alpha'^{-1} \beta$, that is, both $R$ matrices appear in different loops. In this case we have that 
$$\mathbbm E_{L,R}[C_{(\alpha, \beta)}]\leq 1/4D^{2n-| \gamma^{-1} \alpha' \gamma \alpha'^{-1} \beta|} d^{2n- |\tau \alpha|}\leq D^{2n-| \gamma^{-1} \alpha' \gamma \alpha'^{-1} \beta|} d^{2n}.$$ 
Making the change of variables $h =\beta \alpha^{-1}$, $g= \gamma^{-1} \alpha' \gamma \alpha'^{-1}$ that is proven to be one to one in lemma \ref{gamma}, and denoting $h'= h(2n+1)(2n+2)(2n+3)(2n+4)$ we get that
$$\sum_{(\alpha, \beta) \in A_{2n} }  \mathbbm E_{L,R}[C_{(\alpha, \beta)}]\Wg(Dd,\beta \alpha^{-1}) \leq K \sum_{(g, h) \in S_{2n+4}\times S_{2n}} D^{-|  gh'^{-1}|-(3/5)|h|} d^{-(3/5)| h|}=$$
$$=  K \sum_{g\neq \1 \in S_{2n+4}}  D^{-|g|}+K \sum_{h \neq \1, g \in S_{2n+4}}  D^{-|gh|-3/5|h|}  \leq$$
$$\leq K\left(   \sum_{|g|=1}^{2n+3} (\frac {(2n+4)(2n+3)}  {2D})^{|g|}+\sum_{|g|=0}^{2n+3} (\frac {(2n+4)(2n+3)} {2D})^{|g|}\sum_{|h|=1}^{2n-1}(\frac {2n(2n-1)}{2 D^{3/5}})^{|h|} \right) \leq$$
$$\leq K\left(  \frac {(2n+4)(2n+3)}  {2D-(2n+4)(2n+3)}  + \frac {2D}   {2D-(2n+4)(2n+3)}\frac {2n(2n-1)}{2 D^{3/5} - 2n(2n-1)}\right)$$
In the first inequality  we just upper bound the number of permutations with a given number of transpositions, the second is just a geometric sum. As $D \geq n^5$ we get further
$$\sum_{(\alpha, \beta) \in A_{2n} \backslash \{(\1,\1)\} }  \mathbbm E_{L,R}[C_{(\alpha, \beta)}]\Wg(Dd,\beta \alpha^{-1}) \leq  O(D^{-1/5})$$

\item $\alpha,\beta\in \mathcal{B}_{2n} $, where $\mathcal{B}_n$ is the set of tuples where $2n+3$ and $2n+4$ are in the same cycle of the permutation $ \gamma^{-1} \alpha' \gamma \alpha'^{-1} \beta$, that is, both $R$ matrices appear in the same loop. In this case we have that 
$$\mathbbm E_{L,R}[C_{(\alpha ,\beta)}]\leq 1/4D^{2n+1-| \gamma^{-1} \alpha' \gamma \alpha'^{-1} \beta|} d^{2n- |\tau \alpha|}\leq D^{2n+1-| \gamma^{-1} \alpha' \gamma \alpha'^{-1} \beta|} d^{2n}.$$
Applying the same change of variables and consider $\mathcal B'_{2n}$ the image of $\mathcal B_{2n}$ under the change of variables, we get 

$$\sum_{(\alpha, \beta) \in \mathcal B_{2n} }  \mathbbm E_{L,R}[C_{(\alpha, \beta)}]\Wg(Dd,\beta \alpha^{-1}) \leq K \sum_{(g, h) \in \mathcal B'_{2n}} D^{1-|  gh'^{-1}|-(3/5)|h|} d^{-(3/5)| h|}.$$

In order to bound this sum one has to proceed more carefully. The proof follows by bounding independently over the different cases where: $h=\1$, $h=(2n+3,2n+4)$, $h$ is a different transposition, and the rest of terms. For all these cases one has to take into account the properties of the elements in $\mathcal B'_{2n}$, that is, $2n+3$ and $2n+4$ belongs to the same cycle of $gh'^{-1}$ and the parity of $|gh'^{-1}|+|h|$ that is proven in lemma \ref{gamma}. Following this procedure one can prove that 
$$\sum_{(\alpha, \beta) \in \mathcal B_{2n} }  \mathbbm E_{L,R}[C_{(\alpha, \beta)}]\Wg(Dd,\beta \alpha^{-1}) \leq O(D^{-1/5})$$

\end{enumerate}

The result follows joining the two cases and the case $\alpha=\beta=\1$ 

\end{section}

\begin{section}{Proof of Theorem \ref{Lipschitz}}\label{appB}

We use the notation $\| \cdot\|_p$ for the Schatten p-norm. For the first function we have

\begin{equation*}
\begin{aligned}
\|f\|_{Lip}&=\frac {|f(U,V,W,\Lambda, \Omega)-f(U',V',W',\Lambda', \Omega)|} {d((U,V,W,\Lambda),(U',V',W',\Lambda'))} \leq  \frac { |(\trace \rho_l(U,V,W,\Lambda, \Omega))^2-(\trace \rho_l(U',V',W',\Lambda', \Omega))^2|}{d((U,V,W,\Lambda),(U',V',W',\Lambda'))} =\\
&= \frac { |(\trace \rho_l(U,V,W,\Lambda, \Omega)-\trace \rho_l(U',V',W',\Lambda', \Omega))(\trace \rho_l(U,V,W,\Lambda, \Omega)+\trace \rho_l(U',V',W',\Lambda', \Omega))|} {d((U,V,W,\Lambda),(U',V',W',\Lambda'))}\leq\\
& \leq \frac { 2|\trace (\rho_l(U,V,W,\Lambda, \Omega)-\rho_l(U',V',W',\Lambda', \Omega))|} {d((U,V,W,\Lambda),(U',V',W',\Lambda'))}  \leq \frac { 2\trace |\rho_l(U,V,W,\Lambda, \Omega)-\rho_l(U',V',W',\Lambda', \Omega)|} {d((U,V,W,\Lambda),(U',V',W',\Lambda'))}\leq\\
& \leq \frac {2\| V\Lambda V^\dagger U^{ n} (W \Omega W^\dagger \otimes (\ket 0 \bra 0)^{\otimes n}) (U^{ \dagger})^{ n}- V'\Lambda' V'^\dagger U'^{ n} (W' \Omega  W'^\dagger \otimes (\ket 0 \bra 0)^{\otimes n}) (U'^{ \dagger})^{n}\|_1} {{\|U-U'\|_2}+{\|V-V'\|_2}+{\|W-W'\|_2}+{\|\Lambda-\Lambda'\|_\infty}},
\end{aligned}
\end{equation*}
where we are using standard inequalities and for the shake of simplicity we are denoting by the same $U$ unitaries that are acting on different systems and with the same letter $V$ a unitary and its tensor product with the identity. Adding and substracting terms and applying the triangular inequality we get 2n+5 terms of the following form

$$ \frac {2\| V^*\Lambda^* V^{\dagger*} U^{* n} (W^* \Omega W^{*\dagger} \otimes (\ket 0 \bra 0)^{\otimes n}) (U^{* \dagger})^{ n} \|_1} {{\|U-U'\|_2}+{\|V-V'\|_2}+{\|W-W'\|_2}+{\|\Lambda-\Lambda'\|_\infty}} ,$$
where any $X^*$ stands for $X$, $X'$ or $X-X'$ and in any term the latter only appears once. Then

$$\|f\|_{Lip} \leq  \frac {(4n+10)\| V^*\Lambda^* V^{\dagger*} U^{*n} (W^* \Omega W^{*\dagger} \otimes (\ket 0 \bra 0)^{\otimes n}) (U^{* \dagger})^{ n} \|_1} {{\|U-U'\|_2}+{\|V-V'\|_2}+{\|W-W'\|_2}+{\|\Lambda-\Lambda'\|_\infty}}. $$
Applying the inequality $\|XY\|_1\leq \|X\|_1\|Y\|_\infty$  we get
$$\|f\|_{Lip} \leq (4n+10) \frac {\| V^*\|_\infty \|\Lambda^*\|_\infty \|V^{\dagger*}\|_\infty \|U^{* n}\|_\infty  \|W^*\|_\infty \|\Omega\|_1 \|W^{*\dagger}\|_\infty \|(U^{* \dagger})^{ n} \|_\infty} {{\|U-U'\|_2}+{\|V-V'\|_2}+{\|W-W'\|_2}+{\|\Lambda-\Lambda'\|_\infty}}.$$
Now, by the decomposition we have done, any term has only one norm in the numerator of the form of the ones in the denominator. The other norms in the numerator are trivially bounded by one. Thus we get

$$\|f\|_{Lip} \leq 4n+10$$

For the second function we have

\begin{equation*}
\begin{aligned}
\|g\|_{Lip}&=\frac {|g(U,V,W,\Lambda, \Omega)-g(U',V',W',\Lambda', \Omega)|} {d((U,V,W,\Lambda),(U',V',W',\Lambda'))}\leq \\
&\leq  \frac { |\trace (\rho^2_l(U,V,W,\Lambda, \Omega)-\rho^2_l(U',V',W',\Lambda', \Omega))|} {{\|U-U'\|_2}+{\|V-V'\|_2}+{\|W-W'\|_2}+{\|\Lambda-\Lambda'\|_\infty}} \leq 4n+10
\end{aligned}
\end{equation*}
where the result follows using the same techniques.

\end{section}

\end{document}